\documentclass[letterpaper, 10 pt, conference]{ieeeconf}  

\IEEEoverridecommandlockouts                              
\usepackage{tabularx} 
\usepackage{graphicx} 
\usepackage{cite} 
\usepackage[final]{hyperref} 
\usepackage{amsmath} 
\usepackage{amssymb}  
\usepackage{mathtools, cuted,bm}
\usepackage{caption}
\usepackage{subcaption}
\usepackage{esvect}
\usepackage{lipsum, color}

\newcommand{\q}{\mathbf{q}}
\newcommand{\uv}{\mathbf{u}}

\newcommand{\R}{\mathbb{R}}
\newcommand{\I}{\mathbf{I}}

\newcommand{\norm}[1]{\left\lVert#1\right\rVert}

\newcommand{\smallmat}[1]{\left[ \begin{smallmatrix}#1 \end{smallmatrix} \right]}

\newcommand\scalemath[2]{\scalebox{#1}{\mbox{\ensuremath{\displaystyle #2}}}}

\hypersetup{
	colorlinks=true,       
	linkcolor=black,        
	citecolor=black,        
	filecolor=magenta,     
	urlcolor=black         
}

\newtheorem{theorem}{Theorem}
\newtheorem{lemma}[theorem]{Lemma}
\newtheorem{remark}{Remark}
\newtheorem{problem}{Problem}
\newtheorem{proposition}{Proposition}
\newtheorem{definition}{Definition}

\title{\LARGE \bf
Data-Driven Synthesis of Configuration-Constrained Robust Invariant Sets for Linear Parameter-Varying Systems
}

\author{Manas Mejari, Sampath Kumar Mulagaleti, Alberto Bemporad
\thanks{M. Mejari is with IDSIA Dalle Molle
	Institute for Artificial Intelligence, Via la Santa 1, CH-6962 Lugano-Viganello, Switzerland. {\tt\small {manas.mejari@supsi.ch}}} 
\thanks{S. K. Mulagaleti and A. Bemporad are with the IMT School  for  Advanced  Studies  Lucca, Italy. {\tt\small{\{s.mulagaleti, alberto.bemporad\}@imtlucca.it}} }
}

\begin{document}
\newcounter{tempEquationCounter}
\newcounter{thisEquationNumber}
\newenvironment{floatEq}
{\setcounter{thisEquationNumber}{\value{equation}}\addtocounter{equation}{1}
\begin{figure*}[!t]
\normalsize\setcounter{tempEquationCounter}{\value{equation}}
\setcounter{equation}{\value{thisEquationNumber}}
}
{\setcounter{equation}{\value{tempEquationCounter}}
\hrulefill\vspace*{4pt}
\end{figure*}

}
\newenvironment{floatEq2}
{\setcounter{thisEquationNumber}{\value{equation}}\addtocounter{equation}{1}
\begin{figure*}[!t]
\normalsize\setcounter{tempEquationCounter}{\value{equation}}
\setcounter{equation}{\value{thisEquationNumber}}
}
{\setcounter{equation}{\value{tempEquationCounter}}
\end{figure*}

}

\maketitle
\thispagestyle{empty}
\pagestyle{empty}

\begin{abstract}
We present a data-driven method to synthesize \emph{robust control invariant} (RCI) sets for \emph{linear parameter-varying} (LPV) systems subject to unknown but bounded disturbances. A finite-length data set consisting of state, input, and scheduling signal measurements  is used to compute an RCI set and invariance-inducing controller, without identifying an LPV model of the system. We parameterize the RCI set as a \emph{configuration-constrained}  polytope whose facets have a fixed orientation and variable offset. This allows us to define the vertices of the polytopic set in terms of its offset. By exploiting this property, an RCI set and associated vertex control inputs are computed by solving a single \emph{linear programming} (LP) problem, formulated based on a data-based invariance condition and system constraints. We illustrate the effectiveness of our approach via two numerical examples.  The proposed  method can generate RCI sets that are of comparable size to those obtained by a model-based method in which exact knowledge of the system matrices is assumed. We show that RCI sets can be synthesized even with a relatively small number of data samples, if the gathered data satisfy certain excitation conditions.
\end{abstract}

\section{Introduction}\label{sec:introduction}

 Safety guarantees for constrained controlled systems can be analysed through set invariance theory~\cite{fb99}.  
 A \emph{robust control invariant} (RCI) set is a subset of the state-space in which a system affected by bounded but unknown disturbances can be enforced to evolve \emph{ad infinitum}, 
 by an appropriately designed invariance-inducing controller~\cite{fb}. 
 Many works have proposed algorithms for computating such RCI sets along with their associated controllers for \emph{linear parameter-varying} (LPV) systems, see, \emph{e.g.},
\cite{ag19a,Nguyen15, gmfp23auto, mmb23}. These approaches are \emph{model-based},  in that an LPV model of the  system is assumed to be known. However, identifying an LPV model  poses  several challenges~\cite{piga18}.  
Modelling errors can result in the violation of the invariance property and constraints during closed-loop operations. 

To overcome the drawbacks of model-based methods,  data-driven approaches have emerged as favorable alternatives. 
Data-driven \emph{control-oriented} identification algorithms were proposed in~\cite{sam22,chen22} which  simultaneously compute an RCI set and a controller,  while selecting an `optimal' model from the admissible set. The approaches~\cite{sam22,chen22} synthesize RCI sets with reduced conservatism compared to the sequential approach which first selects a model to best fit the data and then computes an invariant set for it.
 Alternatively, \emph{direct} data-driven  approaches were presented in~\cite{bisoffi23,attar23, mejari23, zhong22}, which synthesize RCI sets and controllers directly from  open-loop data,  without the need of model identification.
 The  algorithm presented in~\cite{bisoffi23},
computes a state-feedback controller from open-loop data to
induce robust invariance in a \emph{given} polyhedral set, while methods proposed in~\cite{attar23, mejari23, zhong22}  simultaneously compute  invariance-inducing controllers along with RCI sets having zonotopic~\cite{attar23}, polytopic~\cite{mejari23} or ellipsoidal~\cite{zhong22} representations. 
These contributions, however, are limited to linear time-invariant (LTI) systems. 

For LPV systems,  direct data-driven  algorithms  have mainly focused  on  LPV control design, see,  \emph{e.g.}, 
LPV input-output controllers for constrained systems~\cite{piga18},  predictive controllers~\cite{verhoek21}, and  gain-scheduled controllers~\cite{miller23,verhoek22}. 
 To  our knowledge, only a recent work~\cite{mgp23} has addressed computation of RCI set for  LPV systems in a  data-driven setting. This work differs from~\cite{mgp23} in terms of description of the RCI sets and computational complexity. 
 
 We represent the RCI set with a polytope having fixed orientation and varying offset that we optimize in order to maximize the size of the set. As presented in~\cite{mmb23, villanueva2022configurationconstrained}, we enforce \textit{configuration constraints} (CC) on this polytope, which enable us to switch between their vertex and hyperplane representations. 
We exploit this property to parameterize the controller as a vertex control law which is inherently less conservative than a linear feedback control law~\cite{Gutman1986}.  A single \emph{linear program} (LP) is formulated and solved to compute the CC-RCI set with associated vertex control law, while the approach in \cite{mgp23} requires to solve a semi-definite programming problem.

Our approach does not require an LPV model of the system but only a single state-input-scheduling trajectory consisting of a finite number of data samples. We show via  numerical examples that if the gathered data satisfies certain excitation conditions, then the obtained RCI sets and associated  control inputs can be synthesized with a relatively  small number of  data samples.

\textbf{Paper organization:} The notation and preliminary results used in the paper are given in Section~\ref{sec:notation}. The problem of computing the RCI set from data collected from an LPV system is formalized in Section~\ref{sec:prob}. The configuration-constrained parameterization of RCI sets is presented in Section~\ref{sec:CC_poly}. The proposed data-based invariance conditions and maximization of the size of the set is formulated as an LP in Section~\ref{sec:data-driven invariance}. The effectiveness of the proposed algorithm is demonstrated with two numerical examples in Section~\ref{sec:example}. 

\section{Notations and Preliminaries}\label{sec:notation}
A set of natural numbers between two integers $m$ and $n$, $m\leq n$, is denoted by $\mathbb{I}_m^n \triangleq \{m,\ldots,n\}$.
  Let $A\in \mathbb{R}^{m \times n}$ be a matrix written according to its $n$ column vectors as  $A = \smallmat{a_1 \ \cdots \ a_n}$; we define the vectorization of $A$ as $\Vec{A} \triangleq \smallmat{a^{\top}_{1} \ \cdots \ a^{\top}_{n}}^{\top} \in \mathbb{R}^{mn}$, stacking the columns of $A$. For a finite set $\Theta = \{\theta^{1}, \theta^{2}, \ldots, \theta^{r} \}$ with $\theta^{j} \in \mathbb{R}^{n}$ for $i \in \mathbb{I}_{1}^{r}$, the convex-hull of $\Theta$  is given by,
$\mathrm{ConvHull}(\Theta) \triangleq \left\{ \theta\in  \mathbb{R}^{n}: \theta = \sum_{j=1}^{r} \alpha_j \theta^{j}, \mathrm{s.t} \ \sum_{j=1}^{r} \alpha_j =1, \alpha_j \geq 0  \right\} $. The Minkowski sum of the two sets $\mathcal{X}$ and $\mathcal{Y}$ is defined as $\mathcal{X} \oplus \mathcal{Y}:=\{x+y:x\in\mathcal{X},y\in\mathcal{Y}\}$, and set subtraction as $\mathcal{X} \ominus \mathcal{Y}:=\{x:\{x\} \oplus \mathcal{Y} \subseteq \mathcal{X}\}$.
For matrices $A$ and $B$, $A \otimes B$ denotes their Kronecker product. 
The following results will be used in the paper:
\begin{lemma}[Vectorization] \label{lemma:vectorization}
For matrices $A \in \mathbb{R}^{k \times l}$, $B \in \mathbb{R}^{l \times m}$, $C \in \mathbb{R}^{m \times n}$ and $D \in \mathbb{R}^{k \times n}$,  the matrix equation $ABC = D$ is equivalent to~\cite[Ex. $10.18$]{abadir05},  
    \begin{align}
&(C^{\top} \otimes A) \vv{B} = \vv{ABC} = \vv{D}, \label{eq:vectorize1}
\end{align}
\end{lemma}
  \begin{lemma}[Strong duality]\label{prop:strong_duality}
Given $a \in \mathbb{R}^n$, $b \in \mathbb{R}$, $M \in \mathbb{R}^{m \times n}$ and $q \in \mathbb{R}^m$, the inequality $a^{\top} x \leq b$ is satisfied by all $x$ in a nonempty set $\mathcal{X}:=\{x:Mx \leq q\}$ if and only if there exists some $\bm{\Lambda} \in \mathbb{R}^{1 \times m}_+$ satisfying $\bm{\Lambda} q \leq b$ and $\bm{\Lambda} M=a^{\top}$. 
\end{lemma}


\section{Problem Setting}\label{sec:prob}

\subsection{Data-generating system and constraints}
We consider the following  discrete-time LPV data-generating system
\begin{equation}\label{eq:system}
     x_{t+1}=\mathcal{A}(p_t)x_t\!+\!\mathcal{B}(p_t)u_t\!+\!w_t,
\end{equation}
where  $x_t \in \mathbb{R}^{n}$, $u_t \in \mathbb{R}^{m}$, $p_t \in \mathbb{R}^{s} $, and $w_t \in \mathbb{R}^{n}$ are the state, control input, scheduling parameter, and  (additive) disturbance vectors, at time $t$, respectively.
The  matrix functions $\mathcal{A}(p_t)$ and $\mathcal{B}(p_t)$ 
 have a linear dependency on the parameter $p_t$ as
\begin{align}
\label{eq:linear_dependence}
\mathcal{A}(p_t)=\underset{j=1}{\overset{s}{\scalemath{1}{\sum}}} p_{t,j} A_{o}^j, \quad \mathcal{B}(p_t)=\underset{j=1}{\overset{s}{\scalemath{1}{\sum}}} p_{t,j} B_{o}^j,
\end{align}
where $p_{t,j}$ denotes the $j$-th element of $p_t \in \mathbb{R}^s$ and  $A_{o}^j, B_{o}^j, \ j \in \mathbb{I}_{1}^{s} $ 
are \emph{unknown} system matrices. 
Using \eqref{eq:linear_dependence}, the LPV system~\eqref{eq:system} can be written as
\begin{align}\label{eq:system2}
    x_{t+1} = \underset{M_o}{\underbrace{\begin{bmatrix}
        A_{o}^1 \cdots&A_{o}^s&B_{o}^1 \cdots &B_{o}^s
    \end{bmatrix}}} \begin{bmatrix}
        p_t \otimes x_t \\
         p_t \otimes u_t
    \end{bmatrix} + w_t.
\end{align}
Assume that a  state-input-scheduling trajectory of $T+1$ samples $\{x_t, p_t, u_t\}_{t=1}^{T+1}$  generated from system \eqref{eq:system} is available. The generated dataset is represented by the following matrices
\begin{subequations}\label{eq:data}
\begin{align}
    X^{+} & \triangleq [x_2 \quad  x_3 \quad \cdots \quad x_{T+1}] \in \mathbb{R}^{n\times T}, \\
    {X}^{p}_{u} &\triangleq  \begin{bmatrix} p_1 \otimes x_1 \! & \! p_2 \otimes x_2 \!& \! \ldots & p_T \otimes x_T\\  p_1 \otimes u_1 \! & \! p_2 \otimes u_2 \!& \! \ldots \!& \! p_T \otimes u_T 
    \end{bmatrix} \!\! \in \mathbb{R}^{(n\!+\!m)s\times T}.
\end{align}
\end{subequations}
 Note that the state measurements $x_t$ are generated according to~\eqref{eq:system}, which are affected by  disturbance samples $w_t$ for $t \in \mathbb{I}_{1}^{T+1}$ whose values are \emph{not} known.
However, we assume that for all $t \in \mathbb{I}_1^T$,
\begin{align}\label{eq:disturbance}
    w_t \in \mathcal{W}\triangleq\left \{  w: -h_{n_w} \leq H_ww\leq h_{n_w}\right \},
\end{align}
\emph{i.e.}, the  additive disturbance $w_t$ is unknown but bounded a priori in the  0-symmetric polytope $\mathcal{W}$. Furthermore, we assume that for all $t \in \mathbb{I}_1^T$, the system parameter satisfies
\begin{align*}
    p_t \in \mathcal{P}\triangleq  \mathrm{ConvHull}(\{ p^j \}, j \in \mathbb{I}_{1}^{v_p}),
\end{align*}
where  $\{ p^j \}, j \in \mathbb{I}_{1}^{v_p}$ are $v_p$  given vertices defining the parameter set $\mathcal{P}$. Given the sets $\mathcal{W}$ and $\mathcal{P}$, our goal is to synthesize an RCI set for the LPV system~\eqref{eq:system} that satisfies the  state and input constraints
    \begin{align}\label{eq:constr}
\mathcal{X} \triangleq&\left \{  x: H_xx\leq h_{n_x} \right \}, \  
\mathcal{U} \triangleq \left \{  u: H_uu\leq h_{n_u} \right \},
\end{align}
where $\mathcal{X}$ and $\mathcal{U}$ are given polytopic sets.

\subsection{Set of feasible models}
A set of \emph{feasible models} which are compatible with the measured data $X^{+}, {X}^{p}_{u}$  and the bound on the disturbance samples captured
by the set $\mathcal{W}$, is given as follows
\begin{equation}\label{eq:feasible model set}
   \mathcal{M}_{T} \triangleq \left\{M:
x_{t+1} - M \begin{bmatrix}
p_t \otimes x_t\\ 
p_t \otimes  u_t
\end{bmatrix} \in \mathcal{W},  k \in \mathbb{I}_{1}^{T}   \right \}, 
\end{equation}
where $M= \smallmat{A^1, \ldots A^{s}, \  B^1, \ldots B^{s}} \in \mathbb{R}^{n \times (n+m)s}$ are feasible model matrices. Our assumption is that the true system matrix $M_o$ in \eqref{eq:system2} belongs to this set, $M_o \in \mathcal{M}_{T}$. 
Using the definitions of data matrices in \eqref{eq:data} and disturbance set $\mathcal{W}$ in \eqref{eq:disturbance}, the feasible model set $\mathcal{M}_{T} $ is represented as,
\begin{equation}\label{eq:feasible model set 1}
   \mathcal{M}_{T}\triangleq \left \{M:   
 -h_w \leq H_w X^{+} - H_w M {X}^{p}_{u} \leq h_w  \right \},
\end{equation}
with $h_w \triangleq \smallmat{h_{n_w} \ h_{n_w} \cdots \ h_{n_w}} \in \mathbb{R}^{n_w\times T}$.

We now rewrite the feasible model set $\mathcal{M}_{T}$ in \eqref{eq:feasible model set 1} using the vectorization Lemma~\ref{lemma:vectorization} for $\vv{M} \in \mathbb{R}^{n(n+m)s}$ as
\begin{align}\label{eq:Feasible model set vector}
    \scalemath{0.95}{\mathcal{M}_{T}
    \triangleq \left\{\vv{M} : -\vv{h}_w \!+\!h_M  \leq  H_M \vv{M} \! \leq \!\vv{h}_w\!+\!h_M \right\},}
\end{align}
where we define $H_M \in \mathbb{R}^{Tn_{w} \times n(n+m)s}$, $h_M \in \mathbb{R}^{Tn_w} $ and $\vv{h}_w \in \mathbb{R}^{Tn_w}$ as
\begin{align}\label{eq:Zd}
H_M \triangleq \left( {{X}^{p}_{u}}^{\top} \otimes H_w \right),  \  
  h_M \triangleq \begin{bmatrix}
      H_wx_2 \\ H_wx_3 \\ \vdots \\ H_w x_{T+1}
  \end{bmatrix}, \  
  \vv{h}_w \triangleq \begin{bmatrix}
      h_{n_w} \\ h_{n_w} \\ \vdots \\h_{n_w}
  \end{bmatrix}
\end{align}

\begin{proposition}[Bounded feasible model set]\label{prop:rich_data}
The feasible model set $\mathcal{M}_{T}$ in \eqref{eq:Feasible model set vector} is a bounded polyhedron if and only if $\mathrm{rank}\left( {X}^{p}_{u} \right) = (n+m)s$ and $H_w$ has a full column-rank $n$~\cite[Fact 1]{bisoffi23}. 
\end{proposition}
The full row-rank of ${X}^{p}_{u}$ can be checked from the data,  which also relates to the ``richness" of data  and \emph{persistency of excitation} condition for LPV systems ~\cite[condition 1]{verhoek22}.  If this condition is not satisfied, the set $\mathcal{M}_{T}$ is unbounded which complicates the task of finding a feasible controller and an RCI set for all $M \in \mathcal{M}_T$. 
\begin{remark}
The disturbance set $\mathcal{W}$ in \eqref{eq:disturbance} is assumed to be symmetric, which in turn allows us to check if the model set $\mathcal{M}_{T}$ is bounded  via simple rank condition on the data matrix $X^p_u$, based on~\cite[p. 108]{fb},~\cite[Fact 1]{bisoffi23}. For a general polytope $\mathcal{W}$, the conditions for $\mathcal{M}_{T}$ to be a bounded polyhedron are more involved~\cite[p. 119, ex. 11]{fb}.  
\end{remark}

 \subsection{Invariance condition}
A set $\mathcal{S} \subseteq \mathcal{X} $  is referred to as \emph{robust control invariant} (RCI) for LPV system~\eqref{eq:system2}, if for any given $p \in \mathcal{P}$, there exists a control input $u \in \mathcal{U}$ such that the following condition is satisfied:
\begin{equation}
\label{eq:Invariance Condition}
x\in\mathcal{S}  \;\Rightarrow\; x^+\in\mathcal{S},\; \forall w\in\mathcal{W}, \ \forall M \in \mathcal{M}_{T},
\end{equation}
where the time-dependence of the signals is omitted for brevity and $x^{+}$ denotes the successor state.

We now state and prove  two equivalent conditions for invariance.
Let $\{x^i, i \in \mathbb{I}_1^{v_s}\}$ be the $v_s$ vertices of the convex RCI set $\mathcal{S}$. 
For each vertex  $x^i, i \in \mathbb{I}_1^{v_s}$, we suppose that there exists a vertex control input $\uv^i \in \mathcal{U}, i \in \mathbb{I}_1^{v_s}$. 
\begin{lemma}
If the set $\mathcal{S}$ is robustly invariant for system~\eqref{eq:system2}, then the following two statements are equivalent:
\begin{enumerate}
    \item[$(i)$] for all $x \in \mathcal{S}$, for any given  $p \in \mathcal{P}$, and $\forall (w, M) \in (\mathcal{W}, \mathcal{M}_{T})$,
    \begin{align}
    \label{eq:RCI_condition}
    x^{+} \triangleq M \begin{bmatrix}
p \otimes x\\ 
p \otimes u
\end{bmatrix} + w \in  \mathcal{S};
\end{align}
\item[$(ii)$] for each vertex $\{x^i, \uv^i, i \in \mathbb{I}_{1}^{v_s}\}$, for each  vertex $ \{p^j, j\in \mathbb{I}_{1}^{v_p}\}$ of the set $\mathcal{P}$, and $\forall (w, M) \in \mathcal{W, M}$, \begin{align}
    \label{eq:vertex_RCI_condition}
    {x^{i,j}}^{+} \triangleq M \begin{bmatrix}
p^j \otimes x^i\\ 
p^j \otimes \uv^i
\end{bmatrix} + w \in  \mathcal{S}.
\end{align}
\end{enumerate}
\end{lemma}

\begin{proof}
Since for each vertex $x^i, i \in \mathbb{I}_{1}^{v_s}$ and $p^j, j \in \mathbb{I}_{1}^{v_p}$, it holds that $x^i \in \mathcal{S}$ and $p^j \in \mathcal{P}$, it can be easily seen that $(i) \Rightarrow (ii)$. Now, we prove the converse, \emph{i.e.}, $(ii) \Rightarrow (i)$.
 Any $x \in \mathcal{S}$ can be represented as a convex combination of its vertices as follows:
\begin{align}
    x = \sum_{i=1}^{v_s} \lambda_i x^i, \quad \sum_{i=1}^{v_s} \lambda_i=1, \quad \lambda_i \geq 0 , \forall \ i \in \mathbb{I}_1^{v_s}.
\end{align}
For this state, we choose the corresponding control input as
\begin{align}\label{eq:vertex_control_combination}
    u =  \sum_{i=1}^{v_s} \lambda_i \uv^i.
\end{align}
Note that, $u \in \mathcal{U}$, as $\uv^i \in \mathcal{U}$ and $\mathcal{U}$ is convex. 
Similarly, any given scheduling parameter $p \in \mathcal{P}$ can be expressed as 
\begin{align*}
    p = \sum_{j=1}^{v_p} \alpha_j p^j, \quad \sum_{j=1}^{v_p} \alpha_j=1, \quad \alpha_j \geq 0 , \forall j \in \mathbb{I}_1^{v_p}.
\end{align*}
Applying the control input  \eqref{eq:vertex_control_combination} to System \eqref{eq:system2},  for any $w \in \mathcal{W}$, we get,
\begin{subequations}
    \begin{align}
x^{+} 
&= M \begin{bmatrix}
 \left(\sum_{j=1}^{v_p} \alpha_j p^j\right) \otimes \left(\sum_{i=1}^{v_s} \lambda_i x^i\right)\\ 
\left(\sum_{j=1}^{v_p} \alpha_j p^j\right) \otimes \left(\sum_{i=1}^{v_s} \lambda_i \uv^i\right)
\end{bmatrix} + w,  \\
&= \sum_{j=1}^{v_p} \alpha_j \sum_{i=1}^{v_s} \lambda_i \underset{{x^{i,j}}^{+} \in \mathcal{S}}{\underbrace{\left( M \begin{bmatrix}
 p^j \otimes   x^i\\ 
p^j \otimes  \uv^i
\end{bmatrix} + w \right)}}, \label{eq:xij_plus} \\
&= \sum_{j=1}^{v_p} \alpha_j   \underset{{x^{j}}^{+} \in \mathcal{S}}{\underbrace{ \sum_{i=1}^{v_s} \lambda_i {x^{i,j}}^{+}}},  \label{eq:xj_plus}\\
& =\sum_{j=1}^{v_p} \alpha_j {x^{j}}^{+} \in \mathcal{S}, \label{eq:x_plus}
\end{align}
\end{subequations}
where \eqref{eq:xij_plus} follows from the distributive property of the Kronecker product.  As $\mathcal{S}$ is convex, and from \eqref{eq:vertex_RCI_condition} we know that $ {x^{i,j}}^{+} \in \mathcal{S}$, then ${x^j}^{+} \in \mathcal{S}$ in \eqref{eq:xj_plus}. Similarly, as $x^{+}$ in \eqref{eq:x_plus} is obtained as a convex combination of ${x^j}^{+} \in \mathcal{S}$, it  follows that $x^{+} \in \mathcal{S}$, thus, proving $(ii) \Rightarrow (i)$.
\end{proof}
We now formalize the problem addressed in the paper:
\begin{problem}\label{prob:Problem Formulation}
Given data matrices $(X^+, {X}^{p}_{u})$ defined in \eqref{eq:data} and the  constraint sets \eqref{eq:constr}, compute an invariant set  $\mathcal{S}$  and  associated vertex control inputs $\uv^{i} \in \mathcal{U}, \ i \in \mathbb{I}_{1}^{v_s}$ such that: $(i)$ All elements of the set $\mathcal{S}$  satisfy the state  constraints $\mathcal{S} \subseteq \mathcal{X} $;  $(ii)$ the invariance condition \eqref{eq:vertex_RCI_condition} holds. We also aim at maximizing the size of the RCI set $\mathcal{S}$. 
\end{problem}\vspace{0cm}

\section{RCI set parameterization: configuration-constrained Polytopes}\label{sec:CC_poly}

We parameterize the RCI set $\mathcal{S}$ as the following polytope
\begin{align}
\label{eq:RCI_parameterization}
    \mathcal{S}\leftarrow \mathcal{S}(\q)\triangleq \left\{ x : Cx \leq \q \right\}, \ C \in R^{n_c \times n},  
\end{align}
whose facets have a fixed orientation determined by the user-defined matrix $C$ and  offset $\q \in \mathbb{R}^{n_c}$  to be computed.
We enforce \textit{configuration-constraints} (CC)~\cite{villanueva2022configurationconstrained} over $\mathcal{S}(\q)$,  which enable us to switch between the vertex and hyperplane representation of $\mathcal{S}(\q)$ in terms of $\q$. 

\subsection*{Configuration-constraints}
Given a polytope $\mathcal{S}(\q) \triangleq\{x:Cx \leq \q \}$,  $\q \in \mathbb{R}^{n_c}$, having $v_s$ vertices, the configuration constraints over $\q$ are described by the cone
    \begin{align}
    \label{eq:config_constraint_set}
        \mathbb{S} \triangleq \{ \q: \ E \q \leq \mathbf{0}_{n_c v_s}\}
    \end{align}
with $E \in  \mathbb{R}^{n_c v_s \times n_c}$ whose construction is detailed in Appendix~\ref{sec:appendix}. 
Let $\{ V^i \in \mathbb{R}^{n \times n_c}, i \in \mathbb{I}_{1}^{v_s} \}$ be the matrices  defining the vertex maps of $\mathcal{S}(\q)$, \emph{i.e.},  $\mathcal{S}(\q)=\mathrm{ConvHull}\{V^i \q,  \ i \in \mathbb{I}_1^{v_s}\}$ for a given $\q$. Then, for a particular construction of $\{V^i,i \in \mathbb{I}_1^{v_s}, E\}$, the configuration constraints~\eqref{eq:config_constraint_set} dictate that
    \begin{align}
    \label{eq:config_constraint_relation}
        \forall \q \in \mathbb{S} && \Rightarrow && \mathcal{S}(\q)=\mathrm{ConvHull}\{V^i \q,  \ i \in \mathbb{I}_1^{v_s}\}.
    \end{align}
For a user-specified matrix $C$ parameterizing $\mathcal{S}(\q)$ in \eqref{eq:RCI_parameterization}, we assume we are given  matrices $\{V^i, i \in \mathbb{I}_{1}^{v_s}, E\}$ satisfying~\eqref{eq:config_constraint_relation}. Such matrices  are then used to enforce  that the RCI set $\mathcal{S}(\q)$  is a CC-polytope.   
For further details regarding their  constructions, we refer the reader to Appendix~\ref{sec:appendix}.

    \begin{remark}
The choice of $C \in \mathbb{R}^{n_c \times n}$ acts as a trade-off between representational complexity of the set $\mathcal{S}(\q)$ \emph{vs} the conservativeness of the proposed approach. 
\end{remark}

\section{computation of RCI set  and invariance-inducing controller}\label{sec:data-driven invariance}
In this section, we enforce that the set $\mathcal{S}(\q)$ is RCI under vertex control inputs induced by $\uv^i, i \in \mathbb{I}_{1}^{v_s}$. We recall that a particular construction of matrices  $\{V^i \in \mathbb{R}^{n \times n_c}, i \in \mathbb{I}_{1}^{v_s}, E \in \mathbb{R}^{n_cv_s\times n_c}\}$ satisfying~\eqref{eq:config_constraint_relation} is given.  

We enforce that $\mathcal{S}(\q)$ is a configuration-constrained polytope through the following constraints
\begin{align}
\label{eq:basic_config_con_poly}
    E\q \leq \mathbf{0}.
\end{align}

\subsection{System constraints}
 Let us enforce  the inclusion $\mathcal{S} \subseteq \mathcal{X}$ and input constraints $\uv^i \in \mathcal{U}$. Note that  from~\eqref{eq:config_constraint_relation},  under the constraint \eqref{eq:basic_config_con_poly}, we have the following vertex map of $\mathcal{S}(\q)$, 
 \begin{equation}\label{eq:CC_vertex}
     \mathcal{S}(\q)=\mathrm{ConvHull}\{V^i\q, \ i \in \mathbb{I}_1^{v_s}\}
 \end{equation}
We now enforce the state and input constraints in~\eqref{eq:constr} in terms of $\q$ and $\uv^i$ as follows
\begin{align}
\label{eq:state_input_incl}
    H_xV^i \q \leq h_{n_x}, &&  H_u \uv^i \leq h_{n_u},  && \forall i \in \mathbb{I}_1^{v_s}.
\end{align}

\subsection{Invariance condition}

We now enforce the invariance condition ${x^{i,j}}^{+} \in \mathcal{S}(\q)$ in~\eqref{eq:vertex_RCI_condition} for all $w \in \mathcal{W}$ and for all feasible models in the set $M \in \mathcal{M}_{T}$. Recall that condition~\eqref{eq:vertex_RCI_condition} is enforced at each vertex $\{x^i, \uv^i, i \in \mathbb{I}_{1}^{v_s}\}$ and $\{p^j, j\in \mathbb{I}_{1}^{v_p}\}$ of the set $\mathcal{P}$.

Note that from~\eqref{eq:CC_vertex}, the vertices of $\mathcal{S}(\q)$ are $\{x^i \triangleq V_i \q, \ i \in \mathbb{I}_1^{v_s}\}$, under the  constraints in~\eqref{eq:basic_config_con_poly}. Then, the successor state from ${x^{i,j}}^{+}$ for parameter $p^j$, input $\uv^i$, and disturbance $w$ is given in terms of $\q$ as follows
\begin{align}\label{eq:next_state_vertex}
     {x^{i,j}}^{+} = M \begin{bmatrix}
p^j \otimes V^i \q\\ 
p^j \otimes \uv^i
\end{bmatrix} + w.
\end{align} 
Thus, the inclusion in~\eqref{eq:vertex_RCI_condition} is enforced by the inequality
\begin{align}\label{eq:state_incl_tight}
&C {x^{i,j}}^{+} \leq \q- d  \quad \forall i \in \mathbb{I}_{1}^{v_s}, \ \forall j \in \mathbb{I}_{1}^{v_p}, \ \forall M \in \mathcal{M}_{T},
\end{align}
where $d \triangleq \max\{C w : w \in \mathcal{W}\}$ tightens the set $\mathcal{S}(\q)$ by the disturbance set $\mathcal{W}$.

Using vectorization  in~\eqref{eq:vectorize1}, and substituting \eqref{eq:next_state_vertex},  the inequality~\eqref{eq:state_incl_tight} can be written as follows
\begin{align}\label{eq:state_incl_vec}
   & C \left(\left(\begin{bmatrix}
p^j \otimes V_i \q \\ 
p^j \otimes \uv^i
\end{bmatrix}\right)^{\top} \! \otimes \! I_n \right)\Vec{M} \! \leq \!  \q\!-\!d, \nonumber  \\ 
& \quad \quad \forall \Vec{M} \in \mathcal{M}_{T} \triangleq \{\Vec{M}: H_M \Vec{M} \leq h_M\},  
\end{align}
where we define $\bar{H}_M = \begin{bmatrix}
    H_M \\ -H_M
\end{bmatrix}$ and $\bar{h}_M = \begin{bmatrix}
    \vv{h}_w\!+\!h_M \\ \vv{h}_w\!-\!h_M
\end{bmatrix}$ with $H_M, h_M, \vv{h}_w$  defined as in \eqref{eq:Zd}. 

Using strong duality (Lemma~\ref{prop:strong_duality}), the invariance condition~\eqref{eq:state_incl_vec} holds if and only if
 there exists some  multipliers $\bm{\bm{\Lambda}}^{ij} \in \mathbb{R}^{n_c \times 2Tn_w}_+ $  for all $i \in \mathbb{I}_{1}^{v_s}$, $j \in \mathbb{I}_{1}^{v_p}$ satisfying
 \begin{subequations}\label{eq:invariance_LP}
     \begin{align}
     &\bm{\Lambda}^{ij} \bar{h}_M \leq \q- d, \\
    & \bm{\Lambda}^{ij} \bar{H}_M = C \left(\left(\begin{bmatrix}
p^j \otimes V_i \q \\ 
p^j \otimes \uv^i
\end{bmatrix}\right)^{\top} \otimes I_n \right).
 \end{align}
 \end{subequations}

 \subsection{Maximizing the size of the RCI set}

 We characterize the size of the RCI set $\mathcal{S} \subseteq \mathcal{X}$ as
 \begin{align}
    \label{eq:distance_metric}
    \mathrm{d}_{\mathcal{X}}(\mathcal{S}):=\min_{\epsilon}\{\norm{\epsilon}_1 \ \ \mathrm{s.t.} \ \ \mathcal{X} \subseteq \mathcal{S} \oplus \mathcal{D}(\epsilon)\},
\end{align}
where $\mathcal{D}(\epsilon) \triangleq \{x : Dx \leq \epsilon\}$ is a polytope having user-specified normal vectors $\{D^{\top}_i,i \in \mathbb{I}_1^{m_d}\}$. Thus, we want to compute a desirably large RCI set $\mathcal{S}$ by minimizing the `distance' $\mathrm{d}_{\mathcal{X}}(\mathcal{S})$ in \eqref{eq:distance_metric}. The user-specified matrix $D$ allows us to maximize the size of the set $\mathcal{S}$ in the direction of interest.  

Let $\{\mathrm{y}^l, l \in \mathbb{I}_{1}^{v_x} \}$ be the known vertices of the state-constraint set $\mathcal{X}$, \emph{i.e.}, $\mathcal{X} = \mathrm{ConvHull}\{\mathrm{y}^l, l \in \mathbb{I}_{1}^{v_x} \}$. For each vertex $\mathrm{y}^l$ of $\mathcal{X}$, let $\mathbf{z}^l \in \mathcal{D}(\epsilon)$ and $\mathbf{s}^l \in \mathcal{S}$ for $l \in \mathbb{I}_{1}^{v_x}$ be the corresponding points in the sets $\mathcal{D}$ and $\mathcal{S}$. Then, the inclusion $\mathcal{X} \subseteq \mathcal{S} \oplus \mathcal{D}(\epsilon)$ in \eqref{eq:distance_metric} is equivalent to~\cite{Schneider2013},
\begin{align}
\label{eq:approx_constraints}
\begin{matrix}
\forall  l \in \mathbb{I}_1^{v_x}, \ \exists   \{\mathbf{z}^{l}, \mathbf{s}^{l}\}:   \mathrm{y}^l= \mathbf{z}^{l}+ \mathbf{s}^{l}, \
D\mathbf{z}^l \leq \epsilon, \ C \mathbf{s}^l \leq \q 
\end{matrix}
\end{align}
We now consider the following LP problem which aims at computing the RCI set parameter $\q$ and invariance inducing vertex control inputs $\{\uv^i, i \in \mathbb{I}_{1}^{v_s}\}$ for the LPV system \eqref{eq:system}. Our goal is  to maximize the size of the RCI set $\mathcal{S}(\q)$ (or equivalently, to minimize \eqref{eq:distance_metric}), while satisfying the system constraints, the invariance condition, and the configuration constraints, for all $i \in \mathbb{I}_{1}^{v_s}$, $j \in \mathbb{I}_{1}^{v_p}$ and $l \in \mathbb{I}_{1}^{v_x}$:
 \begin{equation}\label{eq:vol_max_LP}
\begin{array}{lll}
\quad \quad \min & \norm{\epsilon}_1 & \\
{\{\q, \uv^i, \bm{\Lambda}^{ij}, \mathbf{z}^l, \mathbf{s}^l, \epsilon \}} & & \\
\text{subject to:} & \eqref{eq:basic_config_con_poly} & (\text{configuration constraints}), \\
&  
 \eqref{eq:state_input_incl}   &  (\text{state-input constraints}), \\
 &  \eqref{eq:invariance_LP} \;  & (\text{invariance condition}), \\
 & \eqref{eq:approx_constraints} & (\text{volume maximization}). 
 \end{array}
\end{equation}

In terms of computational complexity, the LP in \eqref{eq:vol_max_LP} consists of $n_c v_s$ linear inequalities for expressing the configuration constraints~\eqref{eq:basic_config_con_poly}, $(n_x+n_u)v_s$ linear inequalities for system constraints~\eqref{eq:state_input_incl}, $2n_cv_sv_p$ linear   inequality-equality constraints for invariance \eqref{eq:invariance_LP}, and $v_x(m_d+n_c+n)$ linear equality-inequality constraints for volume maximization. The number of  optimization variables is $(n_c+v_s(m+2v_pn_cTn_w)+2nv_x+m_d)$.   

\subsection{Invariance-inducing controller}

The vertex control inputs $\{\uv^i, i \in \mathbb{I}_{1}^{v_s} \} \in \mathcal{U}$  obtained by solving the LP~\eqref{eq:vol_max_LP} correspond to the vertices $\{x^i , i \in \mathbb{I}_{1}^{v_s} \}$ of the RCI set $\mathcal{S}$. 
Then, for any $x_t \in \mathcal{S}$, an admissible control input $u_t$
can be obtained as follows,
\begin{equation}\label{eq:vertex_control}
    u_t = \sum_{i=1}^{v_s} \lambda^{i,\star}_t \uv^i,
\end{equation}
where $\{\lambda^{i,\star}_t, \ i \in \mathbb{I}_{1}^{v_s}\}$ are computed by solving the following LP:
\begin{equation}\label{eq:control_LP}
\begin{array}{lll}
\{ \lambda^{i,\star}_t \} = &\arg \min  \sum_{i=1}^{v_s} \lambda^{i}_t  & \\
 &\{\lambda^i_t\} & \\
\text{subject to:} & \sum_{i=1}^{v_s} \lambda^i_t x^i = x_t, & 0 \leq \lambda^i_t \leq 1.\\
  \end{array}
\end{equation}
The LP problem~\eqref{eq:vertex_control} is solved at each time step $t$ upon the availability of the new state measurement $x_t$.


\section{Numerical examples}\label{sec:example}

We demonstrate the effectiveness of the proposed approach via two numerical examples. All computations are carried out
on an i7 1.9-GHz Intel core processor with 32 GB of RAM
running MATLAB R2022a.

\subsection{Example 1: LPV Double integrator}
\label{sec:Example1}
We consider the following  LPV double integrator data-generating system~\cite{gmfp23auto},
\begin{align}\label{eq:example1}
    x_{t+1}=\begin{bmatrix} 1+\delta_t & 1+\delta_t \\ 0 & 1+\delta_t \end{bmatrix} x_t+\begin{bmatrix} 0 \\ 1+\delta_t \end{bmatrix} u_t+w_t,
\end{align}
where $|\delta_t|\leq 0.25$, with constraints $\mathcal{X} \triangleq \{x: |x| <=[5\ 5]^{\top} \}$, $\mathcal{U} \triangleq \{u: |u| \leq 1 \} $, and $\mathcal{W}\triangleq \{w:|w| \leq [0.25\ 0]^{\top}\}$.
This system can be brought to the LPV form~\eqref{eq:system} with 
\begin{align}\label{eq:DI_matrices}
\scalemath{0.95}{
    A^1=\begin{bmatrix} 1.25 & 1.25 \\ 0 & 1.25 \end{bmatrix}, A^2=\begin{bmatrix} 0.75 & 0.75 \\ 0 & 0.75 \end{bmatrix}, \begin{matrix} B^1=\begin{bmatrix} 0 & 1.25 \end{bmatrix}^{\top}, \\ B^2=\begin{bmatrix} 0 & 0.75 \end{bmatrix}^{\top}, \end{matrix}}
\end{align}
using  $p_{t,1}=2(0.25+\delta_t)$, $p_{t,2}=2(0.25-\delta_t)$. This corresponds to the simplex scheduling-parameter set $\mathcal{P} = \{p \in \mathbb{R}^2:p \in [0,1], p_1+p_2=1\} = \mathrm{ConvHull}(\smallmat{1\\0}, \smallmat{0\\1})$.
The system matrices in \eqref{eq:DI_matrices} are \emph{unknown} and only used to gather the data. 
A single state-input-scheduling trajectory of $T=100$ samples is gathered by exciting  system~\eqref{eq:example1} with inputs uniformly distributed in $ [-1, 1]$. 
The data satisfies the rank conditions given in Proposition~\ref{prop:rich_data}, \emph{i.e}, $\mathrm{rank}(X^p_u) = (n+m)s= 6$.

We choose matrix $C$ defining an RCI set with representational complexity given by $n_c =50$, \emph{i.e.},
$C \in \R^{50 \times 2}$, such that $\mathcal{S}(\mathbf{1}_{50})$ is an entirely simple polytope. In particular, each row of $C$ is chosen as follows~\cite[Remark 3]{villanueva2022configurationconstrained}
\begin{equation}\label{eq:C_row}
    C^{i} = \left[\cos\left(\frac{2 \pi (i-1)}{n_c}\right), \  \sin\left(\frac{2 \pi (i-1)}{n_c}\right)\right], i \in \mathbb{I}_{1}^{n_c}.
\end{equation}
Based on the selected $C$, we build $\{V^i, i\in \mathbb{I}_1^{50}\}$, and $E$ satisfying the configuration constraints in \eqref{eq:config_constraint_relation}. We refer the reader to Appendix~\ref{sec:appendix} for the details on the construction of $\{V^i, i \in \mathbb{I}_1^{v_s}\}$, and $E$.
We set $D=C$ defining the distance in~\eqref{eq:distance_metric}.  

\begin{figure}[t!]
\centering
    \includegraphics[width=0.9\columnwidth]{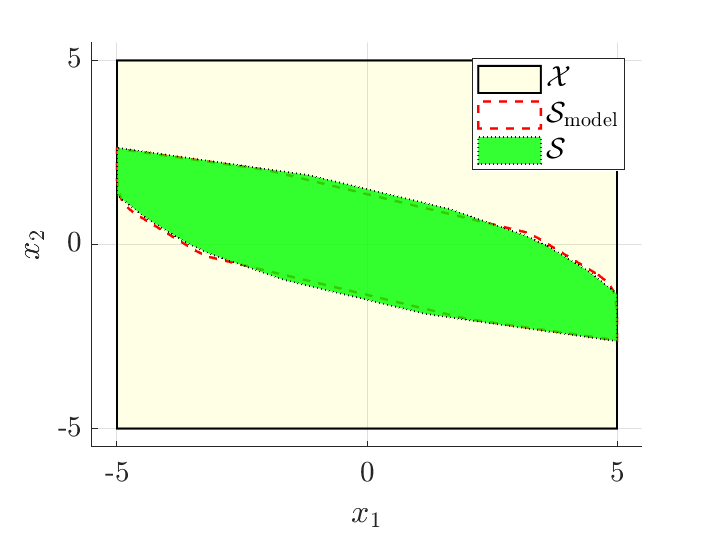}
    \vspace{-0.3cm}
    \caption{Example 1: State constraint set $\mathcal{X}$ (yellow), model-based CC-RCI set $\mathcal{S}_{\mathrm{model}}$ (dashed-red),  data-based CC-RCI set $\mathcal{S}$ (green).}
    \label{fig:rci_plot}
\end{figure}

The RCI set $\mathcal{S}(\q)$ obtained by solving the LP problem~\eqref{eq:vol_max_LP} is shown in Fig.~\ref{fig:rci_plot}.  The total construction and solution time  is $40.6$ s. We compare the proposed approach to a model-based method, where we compute a CC-RCI set $\mathcal{S}_{\mathrm{model}}$ using the knowledge of the true system matrices. In particular, we fix  the model matrix $M$ in  \eqref{eq:vertex_RCI_condition} to the true system matrices $M= \smallmat{A^1, A^2, B^1, B^2}$ given in \eqref{eq:DI_matrices}, and compute $\mathcal{S}_{\mathrm{model}}$ solving an LP minimizing $\mathrm{d}_{\mathcal{X}}(\mathcal{S}_{\mathrm{model}})$. In the model-based case, invariance constraints~\eqref{eq:state_incl_tight} are  directly computed for a given fixed $M$.  
The volume of the RCI set $\mathcal{S}$ obtained with the proposed data-driven proposed algorithm is $25.43$, 
while that provided by the model-based method is $24.56$,
 which shows that the proposed data-based approach generates RCI sets that are of comparable size to those of model-based method.

\begin{figure}[t!]
	\centering
\centering
\includegraphics[width=0.5\columnwidth]{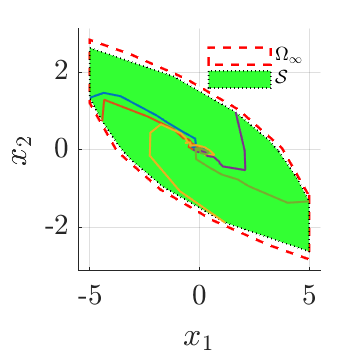}
\includegraphics[width=0.45\columnwidth]{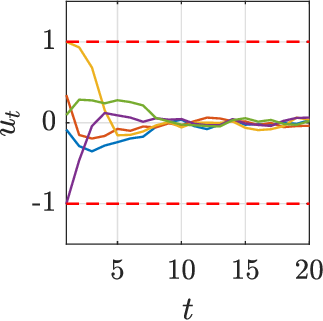}
\caption{Example 1:  Left panel: CC-RCI set $\mathcal{S}$ (green) with closed-loop state trajectories and MRCI $\Omega_{\infty}$ (dashed-red); Right panel: corresponding control input trajectories and input constraints (dashed red).}
\label{fig:closed_loop_eg1}
\end{figure}

Fig.~\ref{fig:closed_loop_eg1}   depicts closed-loop state trajectories starting from some of the vertices of the RCI set (left panel), and corresponding control input trajectories (right panel).  The maximal RCI (MRCI) set $\Omega_{\infty}$ computed using a model-based geometric approach~\cite[Algorithm 10.5]{bbm17}  is also plotted. 
The state trajectories are obtained by simulating the true system~\eqref{eq:example1} in closed-loop with the invariance inducing controller $u_t$ in \eqref{eq:vertex_control} computed by solving the LP \eqref{eq:control_LP} at each time instance. Note that for each closed-loop simulation, a different realization of the scheduling signal $p$ taking values in the given interval $ [0, 1]$ is generated. 
Moreover, during each closed-loop simulation,  different realizations of the disturbance signal $w_t \in \mathcal{W}$  are acting on the system. The result shows that the approach guarantees robust invariance w.r.t. all possible scheduling signals taking values in a given set as well as in the presence of a bounded but unknown disturbance,  while respecting the state-constraints.  The corresponding input trajectories  shown in 
Fig.~\ref{fig:closed_loop_eg1} (right panel)  show that the input constraints are also satisfied.
\begin{table}[]
    \centering
\begin{tabular}{ |c|c|c|c||c||c| } 
 \hline
 $T$  & $30$ & $50$ & $100$ & Model-based &  $\Omega_{\infty}$  \\ 
 \hline
  volume &  22.23 & 24.47 & 25.43 &  24.56 &28.19\\
  \hline
 $\mathrm{d}_{\mathcal{X}}(\mathcal{S}(\q))$ & 168.31 & 166.15 & 164.68 & 162.11&-\\
  \hline
\end{tabular}
    \caption{Example $1$: Size of the RCI set \emph{vs} number of data samples $T$.}
    \label{tab:ex1_vol_vs_T}
\end{table}

Lastly, we analyse the effect of the number $T$ of data samples  on the size of the RCI set. The volume of the RCI set  and the LP objective $\mathrm{d}_{\mathcal{X}}(\mathcal{S}(\q))$ for varying $T=30,50,100$ are reported in Table~\ref{tab:ex1_vol_vs_T}. As $T$ increases, the feasible model set $\mathcal{M}_{T}$ shrinks progressively, $\mathcal{M}_{T+1}\subseteq \mathcal{M}_{T}$, thus constraint $\forall M \in \mathcal{M}_{T}$ is less restrictive, resulting in an increased size of the RCI set.

\subsection{Example 2:  Van der Pol oscillator embedded as LPV}
\label{sec:Example2}

We consider the Euler-discretized LPV representation of the Van der Pol oscillator system~\cite{mmb23} as data-generating system in the form \eqref{eq:system} with 
\begin{align}\label{eq:example2}
   \left[ \begin{array}{c|c}
   A^1 &
   A^2 
\end{array}\right]\!=\!\left[ \begin{array}{c|c}
   \begin{matrix} 1 & T_s \\ -T_s & 1 \end{matrix} &     \begin{matrix} 1 & T_s \\ -T_s & 2 \end{matrix} 
\end{array}\right],  B^{1,2}=\begin{bmatrix} 0 \\ T_s \end{bmatrix}, 
\end{align}
where $T_s=0.1$ is the sampling time.
The scheduling parameters are chosen as $p_{t,1}=1-\mu T_s(1-x_{t,1}^2)$ with $\mu=2$ and $p_{t,2}=1-p_{t,1}$. The system constraints are $\mathcal{X}\triangleq \{x:\norm{x}_{\infty} \leq 1\}$,  $\mathcal{U} \triangleq \{u:|u|\leq 1\}$ and $\mathcal{W}\triangleq \{w:|w| \leq [10^{-3}\ 10^{-3}]^{\top}\}$. 
The scheduling parameter set is $\mathcal{P} \triangleq \{p:p_1 \in [1-\mu T_s,1],p_2 \in [0,1],p_1+p_2=1\}= \mathrm{ConvHull}\left(\smallmat{1\\0}, \smallmat{1-\mu T_s\\\mu T_s}\right)$.
The system matrices $\{A^1, A^2, B\}$ are \emph{unknown} and only used to gather the data. 
A single state-input-scheduling trajectory of $T=100$ samples is gathered by exciting  system~\eqref{eq:example2} with inputs uniformly distributed in $ [-1, 1]$. 
The data satisfy the rank conditions given in Proposition~\ref{prop:rich_data}, \emph{i.e}, $\mathrm{rank}(X^p_u) = (n+m)s= 6$.

 \begin{figure}[t!]
\centering
    \includegraphics[width=0.9\columnwidth]{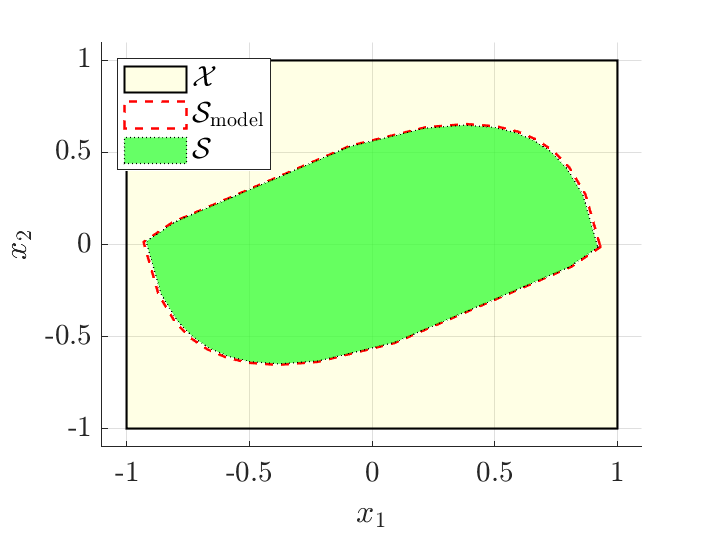}
    \vspace{-0.3cm}
    \caption{Example 2: State constraint set $\mathcal{X}$ (yellow), model-based CC-RCI $\mathcal{S}_{\mathrm{model}}$ (dashed-red), proposed data-driven CC-RCI set $\mathcal{S}$ (green). Note that $\mathcal{S}_{\mathrm{model}}$ and $\mathcal{S}$ are nearly overlapping.}
    \label{fig:rci_plot_vdp}
\end{figure}

\begin{figure}[t!]
	\centering
\centering
\includegraphics[width=0.5\columnwidth]{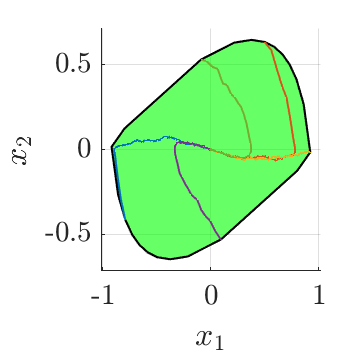}
\includegraphics[width=0.45\columnwidth]{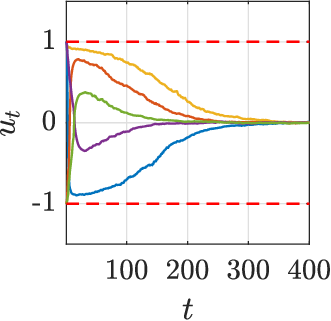}
\caption{Example 2:  Left panel: CC-RCI set $\mathcal{S}$ with closed-loop state trajectories; Right panel: Corresponding control  trajectories and input constraints (dashed red).}
\label{fig:closed_loop_vdp}
\end{figure}

The matrix $C$ parameterzing RCI set is selected with $n_c=30$. Each row of $C \in \R^{30 \times 2}$ is set according to \eqref{eq:C_row}, such that $\mathcal{S}(\mathbf{1}_{30})$ is an entirely simple polytope. Based on the chosen $C$, we build the matrices $\{V^i, i\in \mathbb{I}_1^{30}\}$, and$E$ satisfying the configuration constraints in \eqref{eq:config_constraint_relation}. 
We set $D=C$ defining the distance in~\eqref{eq:distance_metric}.
The RCI set $\mathcal{S}(\q)$ obtained by solving the LP problem~\eqref{eq:vol_max_LP} is shown in Fig.~\ref{fig:rci_plot_vdp}.  The total construction and solution time  is $21.5$ s. 
 For comparision, as in Example $1$, we also compute the CC-RCI set $\mathcal{S}_{\mathrm{model}}$ with the model-based approach using the knowledge of the true system matrices given in \eqref{eq:example2}.
The volume of the RCI set $\mathcal{S}$ with the proposed data-driven algorithm is $1.59$, while that of  $\mathcal{S}_{\mathrm{model}}$ is $1.62$, which shows that the proposed method is able to generate RCI sets that are of similar size to those of the model-based method. 

 Fig.~\ref{fig:closed_loop_vdp}  shows closed-loop state trajectories starting from the vertices of the RCI set (left panel)  for different realizations of the scheduling and disturbance signals during closed-loop simulation. The corresponding invariance-inducing control inputs \eqref{eq:vertex_control} are depicted Fig.~\ref{fig:closed_loop_vdp} (right panel), obtained by solving \eqref{eq:control_LP}, which satisfy the input constraints.

 \begin{table}[ht]
    \centering
\begin{tabular}{ |c|c|c|c||c| } 
 \hline
 $T$  & $20$ & $50$ & $100$ &  $\mathcal{S}_{\mathrm{model}}$  \\ 
 \hline
  volume &  1.50 & 1.56 & 1.59 & 1.62\\
  \hline
 $\mathrm{d}_{\mathcal{X}}(\mathcal{S}(\q))$ & 19.04 & 18.81 & 18..67 & 18.54\\
  \hline
\end{tabular}
    \caption{Example $2$: Size of the RCI set \emph{vs} number of data samples $T$.}
    \label{tab:ex2_vol_vs_T}
\end{table}
Finally, the volume of the RCI set $\mathrm{vol}(\mathcal{S}(\q))$ and the LP objective $\mathrm{d}_{\mathcal{X}}(\mathcal{S}(\q))$ for varying $T=20,50,100$ are reported in Table~\ref{tab:ex2_vol_vs_T}. As $T$ increases, the feasible model set $\mathcal{M}_{T}$ becomes smaller,  resulting in an increased size of the RCI set.
 
 \section{CONCLUSIONS}
 The paper proposed a  data-driven approach  to compute a polytopic CC-RCI set and a corresponding vertex control laws for LPV systems.  The RCI set is parameterized as a configuration-constrained polytope which enables traversing between vertex and hyperplane representation.
 A data-based invariance condition was proposed which utilizes a single state-input-scheduling trajectory without requiring to identify an LPV model of the system. The RCI sets are computed by solving a single LP problem.
The effectiveness of the proposed algorithm was shown via two numerical examples to generate RCI sets from a `small' number of collected data samples. As  future work, we  consider synthesizing parameter-dependent RCI sets for LPV systems in a data-driven setting.

\bibliographystyle{plain}
\bibliography{cc_rci}

\section{Appendix: Configuration-Constrained Polytopes}\label{sec:appendix}
We summarize the main results  from~\cite{villanueva2022configurationconstrained} used in this paper. Let  $\mathcal{S}(\q) \triangleq \{ x\in \mathbb{R}^n:Cx \leq \q\}, \ \q \in \mathbb{R}^{n_c}$. We assume that $\q$ is such that $\mathcal{S}(\q) \neq \emptyset$.
Let  $\mathcal{I}\triangleq \{i_1,\cdots,i_{|\mathcal{I}|}\} \subseteq \mathbb{I}_1^{n_c}$ be the index set
 based on which we define matrices $C_\mathcal{I} \triangleq [C_{i_1}^{\top} \cdots C_{i_{|\mathcal{I}|}}]^{\top} \in \R^{|\mathcal{I}| \times n} $ and ${\q_\mathcal{I} \triangleq [\q_{i_1} \cdots \q_{i_{|\mathcal{I}|}}]^{\top}  \in \R^{|\mathcal{I}|}}$ 
by collecting the rows of  matrix $C$ and elements of vector $\q$ corresponding to the indices in set $\mathcal{I}$. The face of $\mathcal{S}(\q)$ associated with the set $I$ is defined as 
    $\mathcal{F}_\mathcal{I}(\q)\triangleq \{x \in \mathcal{S}(\q): \ C_\mathcal{I} x \geq \q_\mathcal{I}\}$.
\begin{definition}\label{def:entirely_simple}
 A polytope $\mathcal{S}(\q)$ is entirely simple if for all index sets $\mathcal{I}$ such that the corresponding face is nonempty, \emph{i.e.}, $\mathcal{F}_\mathcal{I}(\q) \neq \emptyset$, the condition $\mathrm{rank}(C_\mathcal{I})=|\mathcal{I}|$ holds.  $\hfill\square$
 \end{definition}
 
 For some \emph{given} vector $\sigma \in \R^{n_c}$, suppose that $\mathcal{S}(\sigma)$ is an entirely simple polytope. 
Then, the set of all $n$-dimensional index sets with corresponding faces being nonempty is defined as $\mathcal{V} \triangleq \{ \mathcal{I} : |\mathcal{I}|=n, \ \mathcal{F}_\mathcal{I}(\sigma) \neq \emptyset\}.$
The set $\mathcal{V}$ is the index collection associated with the vertices of $\mathcal{S}(\sigma)$. Let $|\mathcal{V}|=v_s$, \emph{i.e.}, $\mathcal{V}=\{\mathcal{V}_1,\cdots,\mathcal{V}_{v_s}\}$ with $|\mathcal{V}_k|=n$ for each $k \in \mathbb{I}_1^{v_s}$. Then, according to the definition of entirely simple polytopes, $\mathrm{rank}(C_{\mathcal{V}_k})=n$, such that $C_{\mathcal{V}_k}$ is invertible. Let $\I^{n_c}_{\mathcal{V}_k} \in \R^{n \times {n_c}}$ be the matrix constructed using rows of identity matrix $I_{n_c}$ corresponding to indices in $\mathcal{V}_k$. Then, defining the matrices $V^{k}:=C^{-1}_{\mathcal{V}_k} \I^{n_c}_{\mathcal{V}_k} \in \R^{n \times n_c},$ we note that $\left\{V^1 \sigma, \cdots , V^{v_s}\sigma\right\} \in \mathcal{S}(\sigma)$
are the $v_s$ vertices of $\mathcal{S}(\sigma)$. Using matrices $\{V^k, k \in \mathbb{I}_1^{v_s}\}$, define the cone
\begin{align*}
    \mathbb{S}\triangleq \left\{ \q: \underset{E}{\underbrace{\begin{bmatrix} CV^1-I_{n_c} \\ \vdots \\ CV^{v_s}-I_{n_c} \end{bmatrix}}} \q \leq \mathbf{0}\right\},
\end{align*}
which was described in~\eqref{eq:config_constraint_set}. The following result is the basis for the relationship in~\eqref{eq:config_constraint_relation}.
\begin{proposition}{\cite[Theorem 2]{villanueva2022configurationconstrained}}
\label{prop:CC_main}
Suppose that $\mathcal{S}(\sigma)$ is an entirely simple polytope, based on which the vertex mapping matrices $\{V^k, k \in \mathbb{I}_1^{v_s}\}$ and a matrix $E$ defining the cone $\mathbb{S}$ are constructed as discussed above. Then, $\mathcal{S}(\q) = \mathrm{ConvHull}\{V^k \q, \ k \in \mathbb{I}_1^{v_s}\}$ for all $\q \in \mathbb{S} \triangleq \{\q: E\q \leq \mathbf{0}\}$. 
\end{proposition}

\end{document}